\documentclass{article}
\usepackage{spconf,graphicx}
\usepackage{algorithm}
\usepackage{stackengine}
\usepackage{caption}
\usepackage{cite}
\usepackage{color}
\usepackage{xcolor}
\usepackage{amsmath,amsthm,amssymb,amsfonts,bm}
\usepackage{textcomp}
\usepackage{multirow}
\usepackage{subcaption}
\usepackage[]{algpseudocode}
\usepackage{autobreak}

\makeatletter
\newcommand*{\rom}[1]{\expandafter\@slowromancap\romannumeral #1@}

\newtheorem{proposition}{Proposition}
\newtheorem{definition}{Definition}

\renewcommand{\v}[1]{\ensuremath{\boldsymbol{#1}}}

\title{On the privacy of federated Clustering: A Cryptographic View}

\begin{document}
%\ninept
\name{ Qiongxiu Li$^{1}$ and Lixia Luo$^{2}$}
\address{$^{1}$ Tsinghua University, China \\$^{2} $ Huawei Inc. China}
\ninept
\maketitle
\raggedbottom

\addtolength{\abovedisplayskip}{-1.0mm}
\addtolength{\belowdisplayskip}{-1.0mm}

\begin{abstract}
The privacy concern in federated clustering has attracted considerable attention in past decades.  
Many privacy-preserving clustering algorithms leverage cryptographic techniques like homomorphic encryption or secure multiparty computation, to guarantee full privacy, i.e., no additional information is leaked other than the final output.  However, given the iterative nature of clustering algorithms, consistently encrypting intermediate outputs, such as centroids, hampers efficiency. This paper delves into this intricate trade-off, questioning the necessity of continuous encryption in iterative algorithms. Using the federated K-means clustering as an example, we mathematically formulate the problem of reconstructing input private data from the intermediate centroids as a classical cryptographic problem called hidden subset sum problem (HSSP)– extended from an NP-complete problem called subset sum problem (SSP). Through an in-depth analysis, we show that existing lattice-based HSSP attacks fail in reconstructing the private data given the knowledge of intermediate centroids, thus it is secure to reveal them for the sake of efficiency. 
To the best of our knowledge, our work is the first to cast federated clustering's privacy concerns as a cryptographic problem HSSP such that a concrete and rigorous analysis can be conducted.
\end{abstract}
\section{Introduction}
\label{Sect-Intro}
Clustering is a popular unsupervised
learning technique that plays a crucial role in data processing and analysis. By categorizing data into subgroups with similar properties, it provides a powerful means to discover hidden patterns and group similar entities. While clustering techniques are effective, their full potential can only be harnessed when they are applied to comprehensive data sets. In the age of digitalization, datasets often reside across multiple, distributed databases. Federated clustering addresses this challenge by operating directly on decentralized nodes, with each retaining a fragment of the overall dataset \cite{mcmahan2017communication}. However, as data becomes more dispersed, so do concerns about its privacy. The process of clustering, even in a federated setup, poses potential risks of revealing sensitive information. For instance, for certain small or distinct clusters, individual entities might be identifiable. Thus, there is an emerging need for privacy-preserving clustering, which not only allows for decentralized data pattern discovery but also ensures that individual data points remain confidential. 

Existing privacy-preserving federated clustering predominantly fall under two categories: differential privacy (DP)  approaches which insert noise to achieve a rigorous privacy guarantee \cite{stemmer2021locally,balcan2017differentially,su2017differentially} and secure computation approaches \cite{bunn2007secure,mohassel2019practical,kim2018privacy,meng2019private,jha2005privacy,erkin2009privacy,samet2007privacy,liu2014privacy,jiang2020efficient,li2022privacy,sakuma2010large,yuan2017practical} which employ cryptographic tools such as homomorphic encryption (HE) \cite{paillier1999public} or secure multiparty computation (SMPC) \cite{Cramer2015} to allow computation among distrust parties. Typically, DP-based approaches compromise utility for privacy while secure computation requires higher complexity. Comprehensive surveys of these approaches can be found in \cite{meskine2012privacy,hegde2021sok}.

In the realm of clustering algorithms, except the private data,  intermediate outputs such as cluster centroids, cluster assignments, and cluster sizes, emerge potential privacy risks. As highlighted in \cite{hegde2021sok}, it is very difficult to concretely determine whether leaking intermediate information would breach privacy or not, thus it is better to ensure full privacy, i.e.,  no information can be leaked except what the final output. However, full privacy incurs very high complexities as it requires protection for each and every update over the whole iterative process, thereby hindering efficiency and scalability. 
This has spurred the exploration of alternative approaches that improve efficiency by allowing disclosure of some intermediate data, such as cluster centroids \cite{jha2005privacy,erkin2009privacy,samet2007privacy,liu2014privacy,jiang2020efficient,li2022privacy} or the cluster assignments \cite{sakuma2010large,yuan2017practical}. This raises an imperative question:  Is there an inevitable trade-off between privacy and efficiency? Moreover, does the disclosure of intermediate updates gravely jeopardize privacy? To answer these questions, one of the fundamental challenges stems from the intricacies involved in quantitatively assessing the privacy effect of such disclosures. In this paper, we take the first step to conduct a rigorous analysis of this problem in the context of federated K-means clustering. Our main contributions are summarized as follows:
\begin{itemize}
    \item We analyze the privacy effect of revealing intermediate centroids in federated K-means clustering and mathematically formulate it a cryptographic problem known as HSSP. To the best of our knowledge, this is the first work in this context.
    \item With the lattice-based tools tailored for HSSP, we demonstrate a rigorous analysis of the privacy effect. Our findings show that it is very difficult for the adversary to reconstruct the private data given intermediate centroids. This suggests that disclosing intermediate centroids might not inherently jeopardize privacy, thus suggesting that  efficiency may not need to sacrifice privacy. 
\end{itemize}

\section{Preliminaries}
This section reviews necessary fundamentals for the rest of the paper. 

\subsection{Federated K-means clustering}
In a federated setting with $n$ nodes/agents/participants, each and every node of them, denote by $i$, possesses private data $ x_{i}$\footnote{For simplicity, we assume $x_i$ is a scalar and the results can easily be generalized to arbitrary dimensions.}. The goal of federated K-means clustering is to partition the private data $x_i$'s of all nodes into K clusters, the cluster labels are denoted by $\mathcal{K}=\{1,2,...,k\}$, with each cluster being represented by its center $ \mu_{j}, j\in \mathcal{K}$. Let $l_i^{(t)}$  denote the cluster index of node $i$ at iteration $t$, and $\mu_{j}^{(t)}$ denote the center of the $j$-th cluster. Let $t_{\max}$ denote the maximum number of iterations and $\mathcal{T}=\{1,2,\ldots,t_{\max}\}$. For each iteration $t\in \mathcal{T}$, federated K-means clustering consists of two steps:

1) Local cluster assignment:
Every node computes locally, determining which cluster it belongs to based on the distance to the current centroids:
\begin{align}\label{eq.clusterLabel}
 l_i^{(t)}= \arg\min_j \| x_{i}- \mu_{j}^{(t)})\|_2^{2}, j \in \mathcal{K},
\end{align}
where $\|\cdot\|_{2}$ denotes the Euclidean norm. For notation simplicity, we will omit the subscript in subsequent discussions.

2) Centroid  aggregation:
Each node computes its local centroid update. These local updates are then sent to a central coordinator to update the global centroids:
\begin{align}\label{eq.clustercenter}
  \mu_{j}^{(t+1)}= \frac{1}{n_{j}} \sum_{ x_{i} \in {\cal C}_{j}^{(t)} }  x_{i}, j \in \mathcal{K},
\end{align}
where ${\cal C}_{j}^{(t)}=\{ x_{i}| l_{i}^{(t)}=j\}$ denotes the set of all data points within the $j$-th cluster at the $t$-th iteration. The size of this set is given by $n_{j}=|{\cal C}_{j}^{(t)}|$.

\subsection{Threat model}
We assume that the central server is semi-honest (also known as passive or honest-but-curious), it would not change the model parameters maliciously but it can collect information through the learning process. Assume the central server only has the knowledge of all intermediate centroids, its goal is to reconstruct the input data samples given the available information. 

% Quoted in \cite{} that ' it is difficult to concretely determine the effects
% of leaking intermediate information in advance for all
% possible constellations. Hence, privacy research should
% focus on designing efficient private clustering protocols
% that do not leak anything beyond what can be inferred
% from the output, i.e., provide full privacy'

% The main research question is to analyze \emph{whether revealing the intermediate updates such as the intermediate clustering centers will leak the private data in the context of federated K-means clustering?} In other word, we take a step back to question if it is necessary to protect all intermediate updates when considering privacy. 
%the answer to this question is essential as it can significantly reduce the price of encryption 

%Similarly, the final cluster label, denoted as $l_{i}^{f}$, should also be protected as it may reveal information regarding a node, like for example who he/she is in the community. 
%The cluster centers themselves, however, are not considered private as they only represent some group behaviors, not the individual's. 

% \subsection{Problem formualtion}
% We consider the problem \emph{in the context of federated kmeans clustering, whether revealing the intermediate cluster centers, i.e., $\{c_{j}^{(t)}\}_{t\in \mathcal{T}, j \in \mathcal{K}}$ will reveal the input private data? }. %In other words, \emph{Considering privacy, if it is necessary to protect all intermediate cluster centers from being revealed?}

\subsection{Hidden subset sum problem}
We now introduce two classical cryptographic problems that are closely related to federated K-means clustering, as we will show later. 
\begin{definition}\textbf{Subset Sum Problem (SSP)}
Let $n$ be an integer. Given $x_1, \ldots, x_n \in \mathbb{Z}$ and $t \in \mathbb{Z}$, compute $a_1, \ldots, a_n \in \{0, 1\}$ such that
$$
t=a_1 x_1+ a_2 x_2+\cdots+a_n x_n
$$
if they exist.
\end{definition}
Note that SSP is a well-known NP-complete problem, thus currently it is unknown if a deterministic polynomial algorithm exists \cite{dasgupta2008algorithms}.

\begin{definition}\textbf{Hidden Subset Sum Problem (HSSP)}
Let $Q$ be an integer, and let $x_1, \ldots, x_n$ be integers in $\mathbb{Z}_Q$. Let $\v a_1, \ldots, \v a_n \in \mathbb{Z}^m$ be vectors with components in $\{0,1\}$. Let $\v h=\left(h_1, \ldots, h_m\right) \in$ $\mathbb{Z}^m$ satisfying:
\begin{align}\label{eq.hssp}
   \v h= \v a_1 x_1+ \v a_2 x_2+\cdots+\v a_n x_n \quad(\bmod ~Q)  
\end{align}
Given the modulus $Q$ and the sample vector $\v h$, recover the vector $\boldsymbol{x}=\left(x_1, \ldots, x_n\right)$ and the vectors $\v a_i$'s, up to a permutation of the $x_i$'s and $\v a_i$'s.
\end{definition} \label{def.hssp}
%Note that a further extension of HSSP is called Hidden Subset Sum Problem (HLCP) where all setting follows HSSP except that the weight is not long binary but within a larger range. 
We note that HSSP is an extension of SSP, as in SSP, only the weights $a_i$'s are hidden. While in HSSP, both the mixturing weights $\v a_i$ 's and input $x_i$ 's are hidden.  In fact, it is shown that, under certain parameter settings, an HSSP-solver could be used to attack SSP \cite{gini2022hardness}, e.g., when an HSSP instance having a unique solution can be generated, solving such HSSP is sufficient for solving SSP.

% \section{Problem formulation}
% In this section, we formulate the problem and connects it to the above introduced HSSP. 

% \subsection{Connecting to HSSP}
% By inspecting 
% There are two key properties of the formulation problem in \eqref{eq.khsspC}:
% \begin{enumerate}
%     \item The weight matrix $\v A$ is often low rank.
%     \item The weight matrix $\v A$ satisfies \eqref{eq.ksum}.
% \end{enumerate}
% \begin{definition} \textbf{K-HSSP: Hidden Subset Sum Problem for Kmeans clustering}
%\end{definition}

% \section{Hardness of K-HSSP}
% Prove K-HSSP is hard to solve, there are multiple alternatives:
% \begin{enumerate}
%     \item Proving K-HSSP is hard to solve if the weight matrix is low rank, i.e., there are infinitely many solution. (Maybe the partitions and ideas in symmetric polynomials can help)
%     \item Proving K-HSSP is an NPC problem with reduction
%     \item Last choice, prove HSSP can reduce to SSP if there is only one solution (not recommended)
% \end{enumerate}

\section{Lattice-based HSSP attacks}
In this section, we will first introduce the necessary fundamentals of lattices and then introduce the Nguyen-Stern attack (NS attack), which is the very first HSSP attack raised by Nguyen and Stern \cite{nguyen1999hardness}.  Note that there are also two improved attacks based on NS attack, we refer readers to  \cite{coron2020polynomial,coron2021provably} for more details. 
\subsection{Fundamentals of lattices}
A lattice is a discrete subgroup of $\mathbb{R}^m$, which can be defined as follows \cite{cassels1997introduction}.
\begin{definition} \textbf{Lattice}
Given $n$ linearly independent vectors $\v b_1,\dots,\v b_n$ in $\mathbb{R}^m$, the lattice generated by the basis $
\v b_1,\dots,\v b_n$ is defined as the set
\begin{equation*}
\mathcal{L}(\v b_1,\dots,\v b_n)=\left\{\sum_{i=1}^{n}x_i\v b_i|x_i\in \mathbb{Z},i=1,\dots, n\right\},
\end{equation*}
\end{definition}
\noindent where $m$ denotes the dimension of the lattice $\mathcal{L}$, denoted as $\dim (\mathcal{L})$ and $n$ denotes its rank. In this paper, we generally consider integer lattices, i.e., the lattices belonging to $\mathbb{Z}^m$.
\begin{definition} \textbf{Orthogonal lattice}
Let $\mathcal{L}\subseteq \mathbb{Z}^m$ be a lattice. Its orthogonal lattice is given as
$$ \mathcal{L}^{\bot}:=\{\v y\in \mathbb{Z}^m \arrowvert \forall \v x\in \mathcal{L},\langle \v x,\v y\rangle=0\},$$
where $\langle\ ,\ \rangle$ is the inner product of $\mathbb{R}^m$.
Its orthogonal lattice modulo $Q$ is given as
$$ \mathcal{L}_Q^{\bot}:=\{\v y\in \mathbb{Z}^m \arrowvert \forall \v x\in \mathcal{L},\langle \v x,\v y\rangle\equiv 0 \mod Q\}.$$
\end{definition}
%The completion of a lattice $\mathcal{L}$ is the lattice $\bar{\mathcal{L}}=Span_{\mathbb{R}}(\mathcal{L})\cap \mathbb{Z}^m=(\mathcal{L}^{\bot})^{\bot},$ where $Span_{\mathbb{R}}(\mathcal{L})=\{\v b\cdot \v x | \v x\in \mathbb{R}^n\}$ for $\mathcal{L}=\mathcal{L}(\v B).$ 
Note that $(\mathcal{L}^{\bot})^{\bot}$ contains the original lattice $\mathcal{L}$.
%  \begin{theorem}\label{detthm}
%  Let $\mathcal{L}\subseteq  \mathbb{Z}^m$ be a lattice. Then $\dim (\mathcal{L}^{\bot})+\dim (\mathcal{L})=m$.
% \end{theorem}

% \begin{definition} \textbf{Shortest vector}
% Let $\Vert \cdot \Vert$ denote the Euclidean norm. Then we could find  a non-zero vector $\v v$ of the minimal norm in every lattice $\mathcal{L}$. The first minimum of $\mathcal{L}$ is $\lambda_1(\mathcal{L})=\Vert \mathbf{v
% } \Vert,$ and $\v v$ is called the shortest vector.
% \end{definition}

% The following  is a generalization of $\lambda_1$.
% \begin{definition} \textbf{Successive minima}
% Let $\mathcal{L}$ be a lattice of rank $n$. For $i\in \{1,\dots,n\}$, the $i$-th successive minimum $\lambda_i(\mathcal{L})$ is defined by the minimum of maximum norm of any $i$ linearly independent lattice points.
% \end{definition}

The so-called LLL and BKZ algorithms are important building blocks for designing HSSP attacks.
A brief introduction about them is given as follows:
\begin{itemize}
\item \textbf{LLL}:  The LLL algorithm \cite{lenstra1982factoring} is a polynomial time lattice reduction algorithm with a basis of a lattice $\mathcal{L}$ as input and output a basis which is LLL-reduced, meaning that the basis is short and  nearly orthogonal.
\item \textbf{BKZ}: 
The BKZ algorithm \cite{chen2011bkz} is also a lattice reduction algorithm. Increasing the block-size parameter of the algorithm improves the accuracy, at the cost of a longer computation time. Namely, BKZ-2 can produce an LLL-reduced basis in polynomial time, while with full block-size one can retrieve the shortest vector of the lattice in exponential time. 
\end{itemize}

\subsection{NS attack} \label{subsec.ns2step}
The whole process of the NS attack could be divided into two steps:
\begin{itemize}
    \item \textbf{Step 1}: The first step is based on orthogonal lattice attack. As the only known information in \eqref{eq.hssp} is $\v h$, the orthogonal lattice attack starts with computing the orthogonal lattice of $\v h$  modulus $Q$, i.e., $\mathcal{L}_{Q}^{\bot}(\mathbf h)$. As $\v h$ is  a linear combination of the hidden weight vectors $\{\v a_i\}_{i=1}^{n}$, the orthogonal lattice of  $\v A$, i.e., $\mathcal{L}^{\bot}(\v A)$, is then contained in $\mathcal{L}_{Q}^{\bot}(\mathbf h)$. Based on the fact that $\v a_i$'s are binary, the goal is to ensure the first $m-n$ short vectors of the LLL basis of $\mathcal{L}_{Q}^{\bot}(\mathbf h)$ form a basis of $\mathcal{L}^{\bot}(\v A)$.   After obtaining $\mathcal{L}^{\bot}(\v A)$, then compute its orthogonal  $(\mathcal{L}^{\bot}(\v A))^{\bot}$ using the LLL algorithm again, as  $(\mathcal{L}^{\bot}(\v A))^{\bot}$ contains our target lattice $\mathcal{L}(\v A)$. Note that the LLL algorithm is applied twice. The first time is to compute the LLL-reduced basis of $\mathcal{L}^{\bot}_Q(\v h)$; the second time is to compute the orthogonal lattice of $\mathcal{L}^{\bot}(\v A)$ \cite{nguyen2006merkle, chen2018computing}.
    \item \textbf{Step 2}: The second step is to recover the binary vectors $\v a_i$'s from a LLL-reduced basis of $(\mathcal{L}^{\bot}(\v A))^{\bot}$, after that the hidden private data vector $\v x$ can then be recovered. Since the short vectors found by Step 1 may not be short enough, thus the BKZ algorithm is applied in Step 2 to find shorter vectors.  Denote $\{\v v_i\}_{i=1}^{n}$ as the obtained short vectors in  $(\mathcal{L}^{\bot}(\v A))^{\bot}$. Assume that the $\v a_i$'s are the only binary vectors in $(\mathcal{L}^{\bot}(\v A))^{\bot}$, 
     it could be proved with high possibility that $\v v_i$'s are either $\v a_i$'s or $\{\v a_i-\v a_j\}$'s. Thus, take $n$ binary vectors in $\mathcal{L}(\{\v v_i\}_{i=1}^{n})$ as $\v a_i$'s. After obtaining $\v A$, pick up a $n\times n$ sub-matrix $\v A'$ from $\v A$ with non-zero determinant modulus $Q$. Then $\v A'\cdot \v x\equiv \v h' \mod Q$. Hence, the hidden private data can be recovered as $\v x\equiv \v A'^{-1}\v h' \mod Q.$
\end{itemize}

% The strategy described above is guaranteed to work with a good probability if the parameters $m,n,Q$ satisfy certain conditions. The core intuition behind those conditions is that the short vectors of the LLL basis of $\mathcal{L}_{Q}^{\bot}(\mathbf h)$ should form a basis of $\mathcal{L}^{\bot}(\v A)$. Let the vector $\v y$ be  orthogonal modulo $Q$ to $\v h$. We thus have
% \begin{align*}
%     \langle \v y,\v h\rangle=x_1\langle \v y,\v a_1\rangle +\cdots+x_n\langle \v y,\v a_n\rangle \equiv 0 \mod Q,
% \end{align*}
% which implies that the vector $$\mathbf{p}_{\v y}=(\langle \v y,\v a_1\rangle,\dots,\langle\v y,\v a_n\rangle)$$ is orthogonal to the hidden private data vector $\v x=(x_1,\dots,x_n)$, i.e., $\mathbf{p}_{\v y}\in \mathcal{L}_{Q}^{\bot}(\v x)$.
% It follows that if the norm of   $\mathbf{p}_{\v y}$ is less than the first minimum of $\mathcal{L}_{Q}^{\bot}(\v x)$, $\mathbf{p}_{\v y}$ must be a zero vector. Hence, $\v y \in \mathcal{L}^{\bot}(\v a)$. Since $\dim (\mathcal{L}^{\bot}(\v a))=n$, let an upper bound of the $n$ successive minimum of $\mathcal{L}^{\bot}(\v a)$ is less than an estimation of the lower bound of the first minimum of $\mathcal{L}_{Q}^{\bot}(\mathbf x)$. To guarantee $\v y \in \mathcal{L}^{\bot}(\v a)$, the following inequality is required: 
% \begin{align}
% \log Q>\iota mn+\frac{mn}{2(n)}\log m+\frac{n}{2}\log(n), \nonumber
% \end{align}
% where $0<\iota<1$ is decided by the so-called  LLL Hermite factor which controls the quality of the LLL-reduced basis.
% For more details, we refer the readers to\cite{gini2022hardness}.

\begin{table*}
\begin{center}
\begin{tabular}{l|c|c}
\hline
& Random instance of HSSP:$~\v h=\v A \v x$ & K-means instance of HSSP:$~\v c=\v W \v x$  \\
\hline
Adversary's knowledge  &$\v h=\v a_1 x_1+\ldots+\v a_n x_n\in \mathbb{Z}^{m}_Q$ & $\v c=\v w_1 x_1+\ldots+\v w_n x_n\in \mathbb{R}^{m}$\\
\hline
Hidden private data &$\v x\in \mathbb{Z}^{n}_Q$ & $\v x\in \mathbb{R}^{n}$\\
\hline
\multirow{4}{*}{Hidden weight matrix} & $\v A\in \{0,1\}^{m\times n}$ & $\v W\in \{0,1\}^{m\times n}$\\

 & $\operatorname{Rank}(\v A)=n$& $\operatorname{Rank}(\v W)\leq n$\\

& $\mathbb{E}(\|\v a_i\|_{1})=\frac{m}{2}$ & $\|\v w_i\|_{1}=\frac{m}{k}$\\

% & $\forall i\neq j: \mathbb{E}(\|\v a_i-\v a_j\|_{1})\geq \frac{m}{2} $ & $\forall i\neq j:\|\v w_i-\v w_j\|_{1}\||m$\\
% Candidate binary vectors & $\{\v v_i\}\subset (\mathcal{L}^{\bot}(\v A))^{\bot}$ & $\{\v u_i\}\subset (\mathcal{L}^{\bot}(\v W))^{\bot}$\\
\hline
Short vectors & $\{\v v_i\}_{i=1}^{n}\subset (\mathcal{L}^{\bot}(\v A))^{\bot}$ & $\{\v u_i\}_{i=1}^{n}\subset (\mathcal{L}^{\bot}(\v W))^{\bot}$\\
\hline
Recovered binary vectors & $\{ \hat{\v v}_i\}\subset \mathcal{L}(\{\v v_i\}_{i=1}^{n}) \cap \{0,1\}^m$ & $\{\hat{\v u}_i\}\subset \mathcal{L}(\{\v u_i\}_{i=1}^{n})\cap \{0,1\}^m$\\
\hline
% Step 2 output of NS attack& $\{\v v_i, \v v_i \pm \v v_j\}$ & $\{\v v_i, \v v_i \pm \v v_j\}$\\
% \hline
Attack result& $\{\v a_i\}_{i=1}^{n}\cap \{ \hat{\v v}_i\}=\{\v a_i\}_{i=1}^{n}$ & 
$\{\v w_i\}_{i=1}^{n}\cap \{\hat{\v u}_i\}= \emptyset$\\
\hline
\end{tabular}
\caption{Comparisons of random instances and K-means instances of HSSP in terms of a number of key properties.} \label{tab:overview}
\end{center}
\end{table*}

\section{Analysis of NS attack in Federated K-means clustering}
We now delve into the above-described NS attack, focusing on its application in federated K-means clustering. 
We will show that given the knowledge of the intermediate clustering centers  $\v c$ in \eqref{eq.khsspC}, existing HSSP attacks are likely to fail in recovering the hidden private data $x_i$'s. This limitation predominantly stems from the fact that existing HSSP attacks require certain assumptions on the characteristics of the hidden weight vectors. However, such assumption is not satisfied in Federated K-means clustering. 
%To understand the reason why existing attacks fail, we need to first explain a key procedure in \textbf{Step 1} of the above described NS attack. 
\subsection{Privacy of federated K-means clustering as a HSSP}
We first characterize \eqref{eq.clustercenter} with the following representation:
\begin{align}\label{eq.ksum}
\forall j\in\mathcal{K}: c_{j}^{(t+1)}=\v w_{j}^{(t)} \v x
\end{align}
where  $\{\v w_{j}^{(t)}\in \{0,1\}^{1\times n}\}_{t\in \mathcal{T}, j\in \mathcal{K}}$ is referred to as the weight vector characterizing the node selection process, i.e., the $i$-th entry of $\v w_{j}^{(t)}$ is set to $1$ if node $i$ is assigned to cluster $j$, i.e., $l_i^{(t)}=j$, otherwise it is zero. %Note that since every node is expected to belong to a specific cluster, it follows that $\forall i\in \mathcal{T}:~  \sum_{j\in \mathcal{K}} \v w_{j}^{(t)}=\bm 1$.  
$c_{j}^{(t+1)}$ is the sum of all $x_i$'s assigned to cluster $j$. Note that revealing the sum $c_{j}^{(t+1)}$ can provide more information (and hence be riskier) than revealing the average $\mu_{j}^{(t+1)}$. This is because if an adversary knows the number of samples, it can compute the average from the sum, but the reverse isn't necessarily true.  % that $c_{j}^{(t+1)}$'s provide the adversary more information than the average , since the latter also necessitates knowledge of $n_j$. 
Given this, in what follows we will focus on $c_{j}^{(t+1)}$'s and show that even with the knowledge of these values, the private data is protected. %thus is it even more safe for the case of knowing the  average $\mu_{j}^{(t+1)}$.
Stacking the observed information across all clusters and all iterations, i.e.,  $\{c_{j}^{(t)}\}_{t\in \mathcal{T}, j \in \mathcal{K}}$ together we obtain
\begin{align} \label{eq.khssp}
\begin{bmatrix}
        c_{1}^{(1)} \\
       % c_{2}^{(1)}\\
        \vdots  \\
        c_{k}^{(t_{\max})}
    \end{bmatrix}
    = \begin{bmatrix}
        \v w_{1}^{(1)} \\
       % \v w_{2}^{(1)}\\
        \vdots  \\
        \v w_{k}^{(t_{\max})}
    \end{bmatrix}
    \begin{bmatrix}
        x_1 \\
        \vdots \\
        x_n
    \end{bmatrix},
\end{align} 
Further define $\v c= [c_{1}^{(1)},\ldots,c_{k}^{(t_{\max})}]^{T}\in \mathbb{R}^{kt_{\max}}$, $\v W=[\v w_{1}^{(1),T},\\ \ldots,\v w_{k}^{(t_{\max}),T}]^T \in \{0,1\}^{kt_{\max}\times n}$ and $ \v x=[x_{1}, x_{2},..., x_{n}]^T \in \mathbb{R}^{n}$. The above becomes
\begin{align}\label{eq.khsspC}
    \v c=\v W \v x
\end{align}
The problem then becomes \emph{Given the knowledge of $\v c$ in \eqref{eq.khsspC},  can the adversary recover the hidden private data $x_i$'s?}

\subsection{Random instances of HSSP}
Existing HSSP solvers predominantly focus on what is termed ' random instances of HSSP' \cite{nguyen1999hardness,coron2020polynomial,coron2021provably,lyubashevsky2010public}, which is defined as 
\begin{definition}\textbf{A random instance of HSSP} \label{def.rhssp}
Let $x_1, \ldots, x_n$ be integers sampled uniformly at random in $\mathbb{Z}_Q$ and $\v a_1, \ldots, \v a_n \in \mathbb{Z}^m$ be vectors with independent entries sampled uniformly over $\{0,1\}$. They satisfy \eqref{eq.hssp} and the goal is to recover the hidden $x_i$'s and $\v a_i$'s given $Q$ and $\v h$, according to Definition \ref{def.hssp}. 
\end{definition}
\noindent A direct implication of the above definition is that the statistical expectation of the $L_1$ norm of each $\v a_i$ should be $m/2$, i.e., $\forall~i=\{1,\ldots,n\}: ~\mathbb{E}(\|\v a_i\|_{1})=m/2$ where $\|\v \cdot\|_{1}$ denotes $L_1$ norm.

% \begin{assumption}
%  Given a binary matrix $\v A\in \{0,1\}^{m\times n}$, all of its entries are independently sampled uniformly over $\{0,1\}$.
% \end{assumption}

% \begin{assumption}\label{asu.1}
% The hidden weight vectors $\{\v a_i\}_{i=1,\ldots,n}$ are the only binary vectors in the orthogonal lattice $(\mathcal{L}^{\bot}(\v a))^{\bot}$.
% \end{assumption}

Recall that it is explicitly stated in Step 2 of Section \ref{subsec.ns2step}: \emph{Assume that the $\v a_i$'s are the only binary vectors in $(\mathcal{L}^{\bot}(\v A))^{\bot}$, it could be proved with a high possibility that $\v v_i$'s are either $\v a_i$'s or $\{\v a_i-\v a_j\}$'s.} 
We now explain in detail why $\v v_i$'s are likely to be the target hidden weight vectors  $\v a_i$'s themselves or a linear combination of two of them.  The main reason stems from the fact that the probability, of using  linear combination of over two vectors in $\{\v a_i\}$ to construct a vector in $\{-1,0,1\}^{m}$,  is very low.  
Take the linear combination of three vectors as an example, the following result holds \cite{gini2022hardness}.
\begin{proposition}
Let $\{\v a_1, \v a_2,\ldots, \v a_n\}$ be a set of $n$ binary vectors  with dimension $m$ and all entries has equal probability  being $0$ or $1$,  i.e., $\v a_i \in \{0,1\}^{m}$ and  $\mathbb{E}(\|\v a_i\|_{1})=m/2$.
Take three vectors $\v a_i, \v a_j, \v a_k$ and denote $\v v=\v a_i+\v a_j-\v a_k$ as a linear combination of them. 
Thus, the probability that $\v v \in \{-1,0,1\}^m$ is given by  
    \begin{align} \label{eq.svP}
        \operatorname{Pr}\big(\v v\in  \{-1,0,1\}^m \big) = (\frac{7}{8})^m
    \end{align}
For a small $\epsilon\in (0,1)$, a sufficient condition to guarantee the expectation of $\v v \in \{-1,0,1\}^m$ is less than $\epsilon$ is given by 
\begin{align}\label{eq.mnepsilon}
    m\geq 16\log n -6\log \epsilon.
\end{align}
\end{proposition}
\begin{proof}
Since all elements in $\v a_i$'s are either $0$ or $1$, each with $1/2$ probability,  for an arbitrary element in $\v v$, it is of probability $7/8$ that it belongs to $\{-1,0,1\}$. As $\v v$ is of length $m$, the probability of $\v v\in \{-1,0,1\}^m$ is thus $(\frac{7}{8})^m$. Since there are at most $n^3$ triplets of  $\{\v a_i, \v a_j, \v a_k\}$'s, to guarantee $n^3(\frac{7}{8})^m<\epsilon$ it requires $3\log n+m (\log (\frac{7}{8}))<\log \epsilon \rightarrow m> (-3\log n+\log \epsilon)/\log (\frac{7}{8})$. As $-3/(\log \frac{7}{8})<16$ and $1/(\log \frac{7}{8})> -6$ thus \eqref{eq.mnepsilon} holds. 
\end{proof}
%As an example, given $\epsilon=2^{-4}$, for $n\geq 60$, \eqref{eq.mnepsilon} can be satisfied if $m\geq 2n$. For smaller $n$, it requires $m=\max(2n,16\log n +24)$.

Thus, it is of high probability that $\v v_i$'s are either $\v a_i$'s or $\{\v a_i-\v a_j\}$'s. Conversely, $\v a_i$'s can be determined using linear combinations of $\v v_i$'s combined with the fact that $\v a_i$ is binary. This can be formulated as the intersection $\mathcal{L}(\{\v v_i\}_{i=1}^{n}) \cap \{0,1\}^m$. Hence, we conclude that if $\v a_i$ follows Definition \ref{def.rhssp}, it is of high probability that 
\begin{align}\label{eq.vi}
    \{\v a_i\}_{i=1}^{n}\cap \mathcal{L}(\{\v v_i\}_{i=1}^{n}) \cap \{0,1\}^m=\{\v a_i\}_{i=1}^{n}.
\end{align}
%Hence, the hidden vectors $\{\v a_i\}_{i=1}^{n}$ are very likely to be recovered. 

\begin{figure*}[ht]
\begin{subfigure}{0.30\textwidth}
\includegraphics[width=0.9\linewidth]{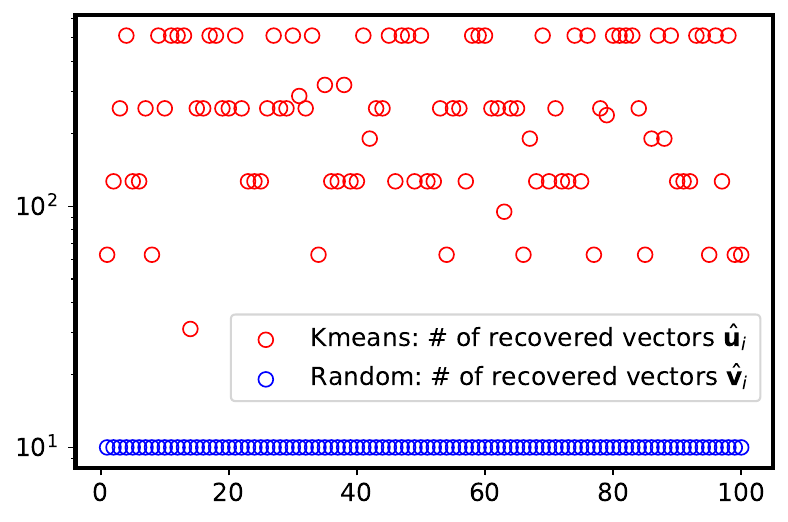} 
\caption{Number of recovered binary vectors: $\hat{\v u}_i$ and $\hat{\v v}_i$  }
%\label{fig:res1}
\end{subfigure}
 \hspace{0.03\linewidth}
\begin{subfigure}{0.30\textwidth}
\includegraphics[width=0.9\linewidth]{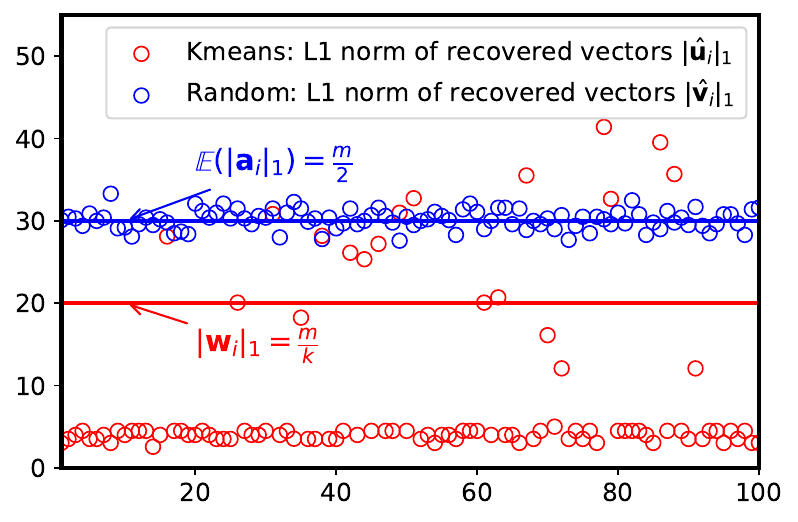}
\caption{Averaged $L_1$ norm of $\|\hat{\v u}_i\|_{1}$, $\|\hat{\v v}_i\|_{1}$  and corresponding true vectors $\|\v w_i\|_{1}$, $\|\v a_i\|_{1}$}
%\label{fig:res3}
\end{subfigure}
 \hspace{0.03\linewidth}
\begin{subfigure}{0.29\textwidth}
\includegraphics[width=0.9\linewidth]{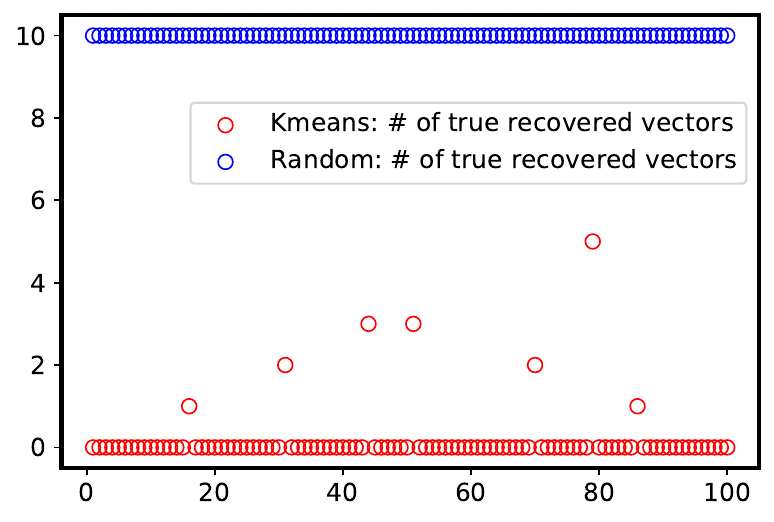}
\caption{Number of true vectors among recovered vectors: $|\{\hat{\v u}_i\} \cap \{\v w_i\}|$ and $|\{\hat{\v v}_i\} \cap \{\v a_i\}|$}
%\label{fig:res2}
\end{subfigure}

\caption{Comparisons of random instance and K-means instance of HSSP in terms of three properties over 100 Monte Carlo runs.}
\label{fig.res}
\end{figure*}

% \begin{figure}[t]
% \centering
% \includegraphics[width=.40\textwidth]{kmeans.png}
% \vspace{-.3\baselineskip}
%  \centering
% \caption{Comparisons of the obtained short vectors for random instance and K-means instance of HSSP.}
% \label{fig:uv}
% \vskip -10pt
% \end{figure}

\subsection{K-means instances of HSSP}
Compare \eqref{eq.khsspC} in K-means clustering to \eqref{eq.hssp}, we can see that the input private data $x_i$'s in the former are real numbers while in the latter they are defined within the modular ring $\mathbb{Z}_{Q}$. However, this discrepancy can be addressed by transforming $x_i$ into $\mathbb{Z}_{Q}$, where floating points can be scaled to integers and the modulo additive inverse can be used for representing negative numbers.  The modulo operation will not make a difference if a sufficiently large $Q$ is considered. We now proceed to discuss the property of the weight matrix $\v W\in \{0,1\}^{m\times n}$ (assuming $m=kt_{\max}$) in K-means clustering. Let $\v w_1,\ldots, \v w_n$ denote the column  vectors of $\v W$. Given that each element in $\v w_i$ signifies  whether the $i$-th input sample $x_i$ is grouped into the corresponding cluster or not, if two samples say $x_i$ and $x_j$ belong to the same cluster at the end of the iteration process, it's unlikely that their cluster assignments would frequently oscillate during the execution of the algorithm. Such result will enforce the following property of the corresponding $\v w_i$ and $\v w_j$:
\begin{align}\label{eq.wij}
    \|\v w_i -\v w_j\|_{1}\ll m/k.
\end{align}
This essentially suggests that the differences in the cluster assignments across all iterations are generally much less than the average length of cluster assignments.

Due to the nature of clustering, the number of clusters is typically much smaller than the number of samples being grouped.  Thus, there will be numerous pairs satisfying \eqref{eq.wij}. A direct implication is that $\v W$ is likely to be rank deficient because of dependencies between $\v w_i$'s.   
We define the K-means variant of HSSP as follows:
\begin{definition}\label{def.khssp} \textbf{A K-means instance of HSSP}
Let $x_1, \ldots, x_n$ be  independently integers sampled in $\mathbb{Z}_Q$. Let $k$ denote the number of clusters and $\v w_1, \ldots, \v w_n \in \{0,1\}^m$ be vectors satisfying $\|\v w_i\|_{1}=m/k$.  They satisfy \eqref{eq.khsspC} and there exists a number of pairs $\v w_i$ and $\v w_j$ satisfying \eqref{eq.wij}. The goal is to recover the hidden $x_i$'s and $\v w_i$'s given $\v c$, according to Definition \ref{def.hssp}.  
\end{definition}
% For two samples $x_i$ and $x_j$ are grouped into one cluster, it follows that the corresponding column vector $\v w_i$ and $\v w_j$ are very similar and satisfy $\|\v w_i-\v w_j\|_{1}\||m$.

% \begin{lemma}
% Let $\{\v w_1, \v w_2,\ldots, \v w_n\}$ be a set of $n$ binary vectors satisfying Definition \ref{def.khssp} . Denote $\v u= \v w_i-\v w_j$ for the pair of $i,j$ satisfying \eqref{eq.wij}, thus 
%     \begin{align}\label{eq.uempty}
%         \{\v u\}\cap \{\v w_1, \v w_2,\ldots, \v w_n\}=\emptyset
%     \end{align}
% \end{lemma}
% The proof is very simple as the length of vectors in $\{\v u\}$ is far less than the length of $\v w_i$'s, which is $m/k$.  As a consequence, it is of low probability that applying NS attack would find out the true hidden weight vectors $\v w_i$'s.  
Let $\{\v u_i\}_{i=1}^{n}$ be the short vectors of $(\mathcal{L}^{\bot}(\v W))^{\bot}$ founded by the BKZ algorithm. We note that it is of high probability that 
 \begin{align}\label{eq.ui}
    \{\v w_i\}_{i=1}^n \cap \mathcal{L}(\{\v u_i\}_{i=1}^{n}) \cap \{0,1\}^m=\emptyset. 
     \end{align}
The rationale is described as follows: since $\mathcal{L}(\v W)\subseteq (\mathcal{L}^{\bot}(\v W))^{\bot}$, it follows that for any pair of $i$ and $j$ where $i\neq j$, $\v w_i-\v w_j \in  (\mathcal{L}^{\bot}(\v W))^{\bot}$. For those pairs $(\v w_i,\v w_j)$ satisfying \eqref{eq.wij}, the lengths of $\v w_i -\v w_j$ and some of  their linear combinations  will be shorter than $\{\v w_i\}_{i=1}^n$. This essentially indicates that there's a plethora of binary vectors within the lattice $\mathcal{L}(\{\v u_i\}_{i=1}^{n})$ that are shorter than $\{\v w_i\}_{i=1}^n$, i.e., \eqref{eq.ui} will hold with high probability.  In the coming section, we will give empirical validations to substantiate this claim.

%(In fact, there exist enough binary short vectors in our expirment) \
% %and our present algorithm could not find $\{\v w_i\}_{i=1}^n$ at all.    
% \begin{lemma}
%  Let $\{\v w_1, \v w_2,\ldots, \v w_n\}$ be a set of $n$ binary vectors satisfying Definition \ref{def.khssp} . Let $\{\v v_i\}_{i=1}^{n}$ be the short vectors of $(\mathcal{L}^{\bot}(\v W))^{\bot}$ founded by the BKZ algorithm. It is of high probability that 
%  \begin{align}\label{eq.ui}
%     \{\v w_i\}_{i=1}^n \cap \mathcal{L}(\{\v u_i\}_{i=1}^{n}) \cap \{0,1\}^m=\emptyset. 
%  \end{align}
% \end{lemma}
% \begin{proof}
% Since $\mathcal{L}(\v W)\subseteq (\mathcal{L}^{\bot}(\v W))^{\bot}$, it follows that for all $i\ne j$, $\v w_i-\v w_j \in  (\mathcal{L}^{\bot}(\v W))^{\bot}$. For those pairs $(\v w_i,\v w_j)$ satisfying $\lVert \v w_i-\v w_j\rVert_1 \ll m/k $, the lengths of $\v w_i -\v w_j$ and some of  their linear combinations  will be shorter than $\{\v w_i\}_{i=1}^n$.  Hence, there are a sufficient number of vectors shorter than $\{\v w_i\}_{i=1}^n$, which completes the proof.%(In fact, there exist enough binary short vectors in our expirment) \
% %and our present algorithm could not find $\{\v w_i\}_{i=1}^n$ at all.    
% \end{proof}
To have a clear overview of random instances and K-means instances of HSSP,  we summarize their key properties in Tab. \ref{tab:overview}.

\section{Experimental validations and discussions}
We now proceed to consolidate our results with experimental validations. We deploy the widely used Iris dataset which comprises 150 samples, each has four attributes and there are three categories of flower types \cite{fisher1936use}. We set the number of clusters $k=3$ and the maximum iteration number as 100. Fig. \ref{fig.res} demonstrates the attack results of both random instances and K-means instances of HSSP over 100 Monte Carlo runs. In each experiment, we randomly subsample $n=10$ samples and $m=60$ observations from the obtained centroids data to ensure the feasibility of completing 100 Monte Carlo runs.
By inspecting plots (a) and (b) we can see that in the K-means case, the algorithm identifies numerous short binary vectors but unfortunately their $L_1$ norm often deviates from the true length $m/3$ given $k=3$. In contrast, in the random instance of HSSP, the $L_1$ norm of discovered short binary vectors $\v v_i$ aligns well with the expected value $m/2$. Consequently, the number of true recovered vectors in the K-means case is often zero out of a total of $n=10$. On the other hand,  for the case of the random instance, a consistent full recovery is observed, as emphasized in plot (c).
Hence, the results in \eqref{eq.vi} and \eqref{eq.ui} are empirically verified, confirming our claim that \emph{given the knowledge of the intermediate centroids, NS attack is not able to the hidden weight vector $\v w_i$'s, thus also private data $x_i$'s}.

Two remarks are in place here. Firstly, the attack results of the other two HSSP attacks \cite{coron2020polynomial,coron2021provably} will also fail in recovering the true hidden vector $\v w_i$'s as there exist lots of false positive binary vectors which are more readily recoverable than the true ones. Secondly, even if $\v w_i$'s are successfully recovered, it is still difficult for the adversary to determine what is the true private data $x_i$'s. This is because there will be infinitely many solutions of $x_i$'s given that the hidden weight matrix $\v W$ is rank-deficient. Hence, overall we conclude that given the knowledge of all intermediate centroids it is infeasible for the adversary to reconstruct the private data $x_i$'s. 
\section{Conclusion}
In this paper we delved deep into the privacy implications of  revealing intermediate centroids in federated K-means clustering. Our novel approach casts this privacy challenge within a cryptographic lens, specifically the hidden subset sum problem (HSSP). Our rigorous analysis suggests that adversaries are unlikely to reconstruct private data given intermediate centroids, which implies their disclosure might not breach privacy.  This insight potentially paves the way for efficient federated clustering algorithms without compromising privacy.

\newpage
\bibliographystyle{IEEEtran}

\bibliography{dualpath}

\end{document}